\documentclass{llncs}

\usepackage{amsfonts}
\usepackage{amsopn}
\usepackage{graphicx}

\begin{document}

\title{Minimal Sum Labeling of Graphs}
%
%
\author{Mat\v{e}j~Kone\v{c}n\'{y} \and Stanislav~Ku\v{c}era \and Jana~Novotn\'{a}
\and Jakub~Pek\'{a}rek \and \v{S}t\v{e}p\'{a}n~\v{S}imsa \and Martin~T\"{o}pfer\thanks{Supported by project CE-ITI P202/12/G061 of GA \v{C}R and grant SVV-2017-260452.}}

%
%
%
\institute{Faculty of Mathematics and Physics, Charles
University, Prague, Czech Republic,\\
\email{matejkon@gmail.com, stanislav.kucera@outlook.com, janka.novot@seznam.cz, edalegos@gmail.com, simsa.st@gmail.com, mtopfer@gmail.com}}
\maketitle              

\begin{abstract}
A graph $G$ is called a {\em sum graph} if there is a so-called {\em sum labeling} of $G$, i.e. an injective function $\ell: V(G) \rightarrow \mathbb{N}$ such that for every $u,v\in V(G)$ it holds that $uv\in E(G)$ if and only if there exists a vertex $w\in V(G)$ such that $\ell(u)+\ell(v) = \ell(w)$. We say that sum labeling~$\ell$ is {\em minimal} if there is a vertex $u\in V(G)$ such that $\ell(u)=1$. In this paper, we show that if we relax the conditions (either allow non-injective labelings or consider graphs with loops) then there are sum graphs without a minimal labeling, which partially answers the question posed by Miller in \cite{b:cocktail} and \cite{b:union}.
\end{abstract}

\section{Introduction}
An undirected graph $G=(V,E)$ is a {\em sum graph} if there exists an injective function $\ell: V(G) \rightarrow \mathbb{N}$ such that every pair of vertices $u\neq v \in V(G)$ is connected via an edge of $G$ if and only if there exists a vertex $w \in V(G)$ such that $\ell(w) = \ell(u) + \ell(v)$. We call the function $\ell$ a {\em sum labeling} or {\em labeling function}. The value $\ell(u) + \ell(v)$ is an {\em edge-number} of the edge $uv$ and it is {\em guaranteed} by vertex~$w$. The {\em sum number}~$\sigma(G)$ of a graph is defined as the least integer, such that $G + \bar{K}_{\sigma(G)}$ ($G$ with $\sigma(G)$ additional isolated vertices) is a sum graph. 


The concept of sum graphs was introduced by Harary \cite{b:harary} in 1990. It was further developed by Gould and R\"{o}dl \cite{b:rodl} and Miller \cite{b:cocktail}, \cite{b:bounds} who showed general upper and lower bounds on $\sigma(G)$ of order $\Omega(|E|)$ for a given general graph $G$ and better bounds for specific classes of graphs.

For some graphs the exact sum numbers are known: 
$\sigma(T_n)=1$ for trees (of order $n\ge2$) \cite{b:trees},
$\sigma(C_n)=2$ for cycles ($n\ge3$, $n\neq4$) and $\sigma(C_4)=3$ \cite{b:harary},
$\sigma(K_n)=2n-3$ for complete graphs ($n\ge4$) \cite{b:complete}, 
$\sigma(H_{2,n})=4n-5$ for cocktail party graphs ($n\ge2$) \cite{b:cocktail} and for complete bipartite graphs \cite{b:bipartite}.

In the work of Miller et al. \cite{b:cocktail} and \cite{b:union}, an open question was raised whether every sum graph has a labeling that uses number 1. Such labelings are called {\em minimal}. In \cite{b:cocktail}, a minimal labeling of complete bipartite graphs is presented. In \cite{b:union}, an upper bound on $\sigma(G)$ for $G$ being a disjoint union of graphs $G_1,G_2\dots,G_n$ is shown. If at least one of the disjoint graphs has minimal labeling then
$$\sigma(G)\le \sum_{i=1}^n\sigma(G_i) - (n-1).$$

In our work, we approach sum graphs from a different perspective. Instead of grounding our research on the properties of graphs, our basic objects are sets of integers. Given an integer set $M$, the rules of sum labeling uniquely define a graph such that it is a sum graph and its labeling consists exactly of all the integers from the given set. From our research we provide two negative answers to questions parallel to the one raised by Miller et al. 

While it is natural to require graphs to have no loops, when we construct a sum graph from an integer set, it seems more natural to allow loops (i.e. for every not necessarily distinct vertices $u,v\in V(G)$ we have $uv\in E(G)$ if and only if there exists $w\in V(G)$ such that $\ell(u)+\ell(v)=\ell(w)$). We call these graphs {\em sum graphs with loops}. When we say {\em sum graphs without loops} we refer to the previous definition, where the vertices are required to be distinct. We show the following: 

\begin{theorem}\label{thm:loops}
There exists an infinite family of sum graphs with loops which admit no minimal labeling.
\end{theorem}

Another relaxation of the original problem (without loops) is to allow the integer set $M$ to be a multiset. This of course causes the labeling function to cease being injective, thus we call such graphs \textit{non-injective sum graphs}. Nevertheless, in this approach we may consider graphs without loops and obtain the following similar result: 

\begin{theorem}\label{no1onoinj}
There exists an infinite family of non-injective sum graphs (without loops) which admit no minimal labeling.
\end{theorem}

\subsection{Preliminaries}

Let $G$ be a sum graph with some labeling $\ell$. We call set~$M = \{\ell(v) : v \in V(G)\}$ {\em label-set} of $G$.

For a finite multiset of natural numbers $S\subset \mathbb N$, let $G_S$ be the graph with elements of $S$ being its vertices and for every $u,v\in S$ let there be an edge $uv\in E(G_S)$ if and only if $u+v\in S$. Depending on context, we sometimes allow $G_S$ to have loops. We say that set~$S$ {\em induces} a graph $G_S$. Let $f$ denote the natural bijection of vertices of $G_S$ and integers in $S$. For any integer $i \in S$ we denote $\psi(i)$ the subset of $V(G_S)$ such that a vertex $v\in \psi(i)$ if and only if $f(v)=i$. We say that an integer $i$ \textit{induces} a vertex $v$ if $v \in \psi(i)$. 

We say that two vertices $u,v$ of a graph $G$ without loops are {\em equivalent} if it holds that $N(u) \backslash \{v\} = N(v) \backslash \{u\}$.

\begin{lemma}\label{NonEq}
Let $G$ be a graph with labeling $\ell$ and let $u,v$ be two of its vertices. If $\ell(u) = \ell(v)$, then $u$ and $v$ are equivalent. 
\end{lemma}
\begin{proof}
Suppose $\ell(u) = \ell(v)$. Consider any $w \in V(G)$ other than $u$ and $v$. The edge-numbers of $wv$ and $wu$ are the same, so either both edges are present or none of them is. Since we do not consider loops, the only remaining edge to consider is $uv$. If $uv \notin E(G)$, then clearly $N(u) = N(v)$. If $uv \in E(G)$, then all neighbors of $v$ are also neighbors of $u$ except $u$ itself and vice versa. In both cases, $u$ and $v$ are equivalent. 
\end{proof}

To get an analogous definition for a graph with loops one would not exclude the vertices $v$ and $u$ from the neighborhoods. This rather subtle difference actually makes dealing with sum graphs with loops much easier.

We define an operation that removes loops from sum graphs. Consider a sum graph $G$ with loops. Let us take one by one each vertex $v$ with a loop and choose any integer $k \geq 2$ (independently for each vertex). We replace $v$ with a clique $K_k$ and connect all neighbors of $v$ to all vertices from $K_k$. We denote set of all possible results of this operation by $\mathcal{C}(G)$, and denote $\mathcal{C}^{k}(G)$ the unique result where we fix all values $k$ from the construction to a given fixed value. Let $G$ be induced by a multiset $M$, we may equivalently define $\mathcal{C}(G)$ as all graphs induced by all possible multisets obtained from the set $M$ via raising the multiplicity of membership of any $i \in M$ such that $2i \in M$.

\section{Sum Graphs with Loops}\label{sec:loops}
This section serves as an introduction to the topic of this paper. We show that there is a sum graph with loops that admits no minimal sum labeling. Though this is weaker result than Theorem~\ref{thm:loops}, we provide a direct proof without usage of complex tools. The full proof of Theorem~\ref{thm:loops} is given later as a consequence of Theorem~\ref{no1onoinj}. 

For a proof, we use the graph induced by the set $\{2,3,4,6,7\}$. This specific graph was chosen based on a result of a computer experiment as the smallest graph induced by an arithmetic sequence with difference 1 starting from 2 with one element missing such that no minimal labeling was found. The proof goes through several cases and is given in the Appendix.

\begin{theorem} \label{No1wlEx}
There exists a sum graph with loops that admits no minimal labeling. 
\end{theorem}
\begin{figure}
\centering
\includegraphics{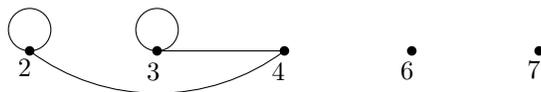}
\caption{The sum graph induced by the set $\{2,3,4,6,7\}$.}
\end{figure}

Although Theorem \ref{No1wlEx} does not require the labeling function to be injective, the theorem holds also under the constraint of injectivity. The labeling used to induce the graph is certainly injective, and the theorem shows that there exists no minimal labeling; thus, in particular no injective minimal labeling. 

It is easy to observe that sets $\{1,2,...,k\}$ induce graphs with the maximum number of edges out of all sum graphs with loops on the same number of vertices. In fact, it can be shown that any graph with the same number of vertices and edges is necessarily isomorphic to this graph. Based on this observation, it seems reasonable to assume that the set $\{2,3,...,k-1,k\}$ might induce a graph with specific structure and possibly exclude all labelings with 1 as 1 was removed from the inducing set. This idea does not hold, as the same graph is induced for example by the set $\{1,2,...,k,3k\}$. 

However, the situation seems to change dramatically once we remove one more value. Let us call {\em gap-graphs} of size $k$ all graphs induced by the set obtained from $\{2,...,k\}$ by removing one element we call a {\em gap}. Hence, a gap-graph of size $k$ has $k-2$ vertices. While for small values of $k$ some gap-graphs have a minimal labeling, we conjecture that for $k \geq 10$ none of the gap-graphs of size $k$ with gap $i$ such that $3 < i < k$ has a minimal labeling. While we do not prove this conjecture, it serves as a basic inspiration for our main result and as a consequence of our main result, we prove Theorem~\ref{thm:loops} by providing a partial proof of the conjecture for sufficiently large $n$ and a specific choice of gap.

\section{Non-injective Sum Graphs without Loops}\label{sec:noloops}
In this section, we construct graphs with (non-injective) sum labelings such that they do not admit minimal labelings.

Based on the result from the previous section and the conjecture about gap-graphs, it is a natural question whether we can modify the gap-graph idea to remove all loops and keep the desired properties that may prevent existence of minimal labelings. Let us consider the gap-graph $G$ induced by the set $\{2,3,4,6,7\}$ from Theorem \ref{thm:loops}, and its loopless modification $H = \mathcal{C}^2(G)$. The sum graph $H$ is induced by a multiset $\{2,2,3,3,4,6,7\}$ by definition. Unfortunatelly a graph isomorphic to $H$ is induced by a multiset $\{1,5,2,2,4,6,9\}$ (with the same ordering of vertices). This numbering can be naturally extended to a (non-injective) labeling of any graph from $\mathcal{C}(G)$. While this gives us a negative result, we show that for large enough gap-graphs and at least some choice of gap, the construction $\mathcal{C}$ does in fact guarantee all produced graphs to admit no minimal labeling. 

We first develop some tools applicable (up to minor adjustments) to all flavors of sum graphs (injective, non-injective, with or without loops). Namely, that labels of each sum labeling of a sum graph can be described as a set of arithmetic sequences. We show how several simultaneous description of this form limit each other and in doing so reflect some structural aspects of underlying sum graphs directly into their label-sets. 

The graphs we work with are based on arithmetic sequence of integers. We define graph $A_n$ as the unique sum graph with loops induced by the set $\{2,3,...,n-1,n,n+2\}$, in other words by a set of all integers from 2 up to $n+2$ without the second-last value $n+1$. Note that graph $A_n$ has exactly $n$ vertices.

\subsection{Sequence Description}

Let us denote an arithmetic sequence $(a,a+d,a+2d,...,a+jd)$ with difference $d$ as $(id+a)_i$ where we always consider $i$ going from 0 up to some unspecified integer. We also generally refer to sequences with difference $d$ as {\em $d$-sequences}. 

Let us fix a vertex $v$ and call it {\em generator}. The set of {\em terminals} associated with this generator is defined as $V(G)\backslash N(v)$. Note that $v$ has exactly $n-deg(v)$ terminals. It is an important observation is that unless $v$ has a loop, it is its own terminal. We say that a terminal $w$ is a {\em proper terminal} of $v$ if $w \neq v$, and is \textit{improper terminal} otherwise. 

\begin{lemma}\label{SeqProper}
Let $G$ be a sum graph without loops, let $\ell$ be a fixed labeling of $G$, let $v$ be a fixed generator and let $w$ be a proper terminal of $v$. Then there is no vertex in $G$ labeled $\ell(v)+\ell(w)$. 
\end{lemma}
\begin{proof}
For contradiction, let a vertex $u \in V(G)$ be labeled $\ell(v) + \ell(w)$. Then the edge-number of the edge $vw$ is guaranteed by $u$, thus $vw \in E(G)$. This however makes $w$ a neighbor of $v$ and we reach a contradiction. 
\end{proof}

The previous lemma does not hold for improper terminal $v$, as the edge used in the proof would be a loop and thus would not be an edge of $G$ even though it is technically guaranteed. In case of sum graphs with loops this issue does not arise. 

\begin{lemma}\label{SeqSequences}
Let $G$ be a sum graph with fixed labeling $\ell$, let $M$ denote the label-set of $G$, let $v$ be a fixed generator and let $u$ be a non-terminal of $v$. Then there exists a sequence $S = (i \ell(v)+\ell(u))_i \subseteq M$ such that its last element is a label of a proper terminal of $v$. 
\end{lemma}
\begin{proof}
Since $u$ is a non-terminal of $v$, there exists an edge $uv$ with edge-number $\ell(v)+\ell(u)$. Clearly, this edge is guaranteed by some vertex $u_1$ labeled $\ell(v)+\ell(u)$. The vertex $u_1$ cannot be $v$ as $\ell(u_1) > \ell(v)$. If $u_1$ is a proper terminal, then $S = (\ell(u), \ell(u)+\ell(v))$ and we are done. Otherwise, $u_1$ is a non-terminal and we iterate the previous argument, building a sequence $S'$ from $u_1$ and setting $S = (\ell(u)).S'$ where the dot operation denotes sequence concatenation.  
\end{proof}

\begin{lemma}\label{SeqCover}
Let $G$ be a sum graph, and let $v$ be a fixed generator with $k$ terminals and label $\ell(v) = g$. Then for every labeling $\ell$ and associated label-set $M$, the following holds: 
\begin{enumerate}
\item The label-set $M$ can be described as a union of at most $k$ distinct arithmetic sequences with difference $g$. 

\item The last element of each of the sequences is a label of a terminal. 

\item Each label of a proper terminal is the last element of one of the sequences. 

\item If $v$ has $j$ non-equivalent terminals and the description of $M$ has $j$ sequences, then one of the sequences is a singleton sequence $(g)$. 
\end{enumerate}
\end{lemma}
\begin{proof}
Let $u$ be a non-terminal of $v$ such that it has the lowest label of all non-neighbors. From Lemma \ref{SeqSequences} we have a sequence $S_1$ such that it covers some elements of $M$ and ends in a terminal. If $M \backslash S$ contains non-terminals, we iterate by using the Lemma \ref{SeqSequences} with the lowest remaining non-terminal label. If $M \backslash S$ contains no further non-terminals and is non-empty, then we create a singleton sequence $(\ell(w))$ for each proper terminal $w$. Clearly the only element of $M$ that can remain not-covered is $\ell(v)$, as $v$ is improper terminal. In such case we create one more singleton sequence $(\ell(v))$. Suppose $S_1 = (ig+a)_i$ and $S_2 = (ig+b)_i$, for some integers $a,b$, are two constructed non-singleton sequences ($S_1$ was constructed first). Since both $S_1$ and $S_2$ have the same difference, $a \notin S_2$ and $a < b$ by the choice of $a$ and $b$, both sequences are distinct. Clearly, no singleton sequence shares an element with any other sequence. Since there are exactly $k$ terminals, each sequence ends with a terminal label and all sequences are distinct, we get that there are at most $k$ sequences in total. Naturally, $M$ is in the union of all of the constructed sequences and the sequences contain no extra elements. This proves points 1 and 2. 

From Lemma \ref{SeqProper}, we have that if a sequence contains the label of a proper terminal, the hypothetical next element of the sequence is not a label of any vertex. Thus labels of proper terminals can only be the last elements in sequences, which proves the point 3. 

Let $v$ have $j$ non-equivalent terminals. From Lemma \ref{NonEq} we have that there are $j$ distinct labels of terminals in $M$. If there are $j$ sequences in the description of $M$ via generator $v$, we deduce from the previous points that each of the distinct labels is the last element of a distinct sequence. In particular, the label $g$ must be the last in a sequence as $v$ is its own (improper) terminal. Since all labels are positive and the difference of each sequence is $g$, the label $g$ must form a singleton sequence $(g)$. 
\end{proof}

Note that some labels generated by a sequence may be labels of several vertices, if these are equivalent. 

Let $v$ be a vertex of a fixed graph $G$. Let \textit{$v$-cover} denote the set of sequences covering the label-set $M$ of $G$ as described in Lemma \ref{SeqCover}. Each such set is associated with one difference value. If $\alpha$ is this difference value (i.e. $\alpha = \ell(v)$), then we also reference to such a cover as \textit{$\alpha$-cover}.

\subsection{Cover Merging}

While the results of the previous section do not give particularly strong results when all the degrees are low (in respect to the number of vertices), it does give strong limits on potential labelings once we have one or more vertices with almost full degree. In this section, we expand our tools to impose additional constraints when applying the previous results to multiple vertices simultaneously. 

Let $G$ be a graph with fixed proper labeling and let $M$ be its label-set. Then each vertex $g_i$ induces a cover $C_i$ of $M$, as described in Lemma \ref{SeqCover}. Each such cover $C_i$ is associated with a difference $d_i$ ($d_i = \ell(g_i)$) and a number of terminals of $g_i$ (including $g_i$) denoted $t_i$. 

We say that two covers $C_i$ and $C_j$ are \textit{mergeable} if the number of non-equivalent vertices in $G$ (and thus also the size of any label-set of $G$) is at least $2 \cdot (t_i - 1) \cdot (t_j - 1) + 3$. A proof of the following lemma is given in the Appendix.

\begin{lemma}\label{mergeability}
Let $C_i$ and $C_j$ be mergeable covers, then there exists a pair of sequences $S_i,S_j$ from covers $C_i$ resp. $C_j$ such that they share at least three elements. 
\end{lemma}

\begin{lemma}\label{2MergeEq}
For any fixed mergeable covers $C_1,C_2$ there exist positive integers $j,k$, such that $GCD(j,k) = 1$, $j \leq t_1 - 1$, $k \leq t_2 - 1$ and the equation $k \cdot d_1 = j \cdot d_2$ holds. 
\end{lemma}
\begin{proof}
Let us fix two mergeable covers $C_1$ and $C_2$. Let $S_1,S_2$ denote sequences such that $S_1 \in C_1$, $S_2 \in C_2$ and $|S_1 \cap S_2| \geq 3$. Let $x_0,x_1,x_2$ be the three smallest elements of $S_1 \cap S_2$ so that $x_0 < x_1 < x_2$. Since all of them come from both sequences with no holes, we may denote the distance between elements $m = x_1 - x_0 = x_2 - x_1$. 

As both $x_1,x_2$ belong to both sequences, it must hold that $m = k \cdot d_1 = j \cdot d_2$ for some positive integers $j,k$ such that $x_2$ is $k$-th element following $x_1$ in $S_1$ and also $j$-th element following $x_1$ in $S_2$. 

Suppose $d_1 = d_2$, then $k = j = 1 \leq t_1 - 1,t_2 - 1$ and the lemma holds trivially. We may assume $d_1 \neq d_2$. From definition of a single sequence, $S_1$ contains all $k$ possible elements from interval $(x_1,x_2]$. Similarly, $S_2$ contains all $j$ elements from $(x_1,x_2]$. If $GCD(k,j) > 1$ then there is some $m_0 < m$ such that $x_1+m_0$ is an element of both sequences which would contradict the minimality of $x_1,x_2$. More generally, if any two elements $x_3,x_4$ of $S_1$ such that $x_1 < x_3 < x_4 \leq x_2$ belonged to the same sequence $S_0 \in C_2$, then $x_1 + (x_4 - x_3) \in S_1 \cap S_2$ which would contradict the minimality of $x_2$. Analogously, we would reach a contradiction if any such $x_3,x_4$ belonged to any $S_0' \in C_1$. Thus the $k$ elements of $S_1$ from the interval $(x_1,x_2]$ fall into distinct sequences from $C_2$ and we have $k \leq t_2$ as $t_2$ limits the number of sequences in $C_2$. Analogously we get $j \leq t_1$. 

From Lemma \ref{SeqCover}, we know that if the $k$ elements of $S_1$ fall into $t_1$ distinct sequences from $C_2$, then one of them has to be the singleton sequence $(d_2)$. Recall that we chose $x_1,x_2$ from $S_1 \cap S_2$ so that they are preceded by some element $x_0$ from $S_1 \cap S_2$. This means that $x_1 > d_2$ and thus the singleton sequences from $C_2$ cannot, in fact, play any role. Thus, the limit on the number of sequences involved can be further reduced to $k \leq t_2 - 1$. Symmetrically, we obtain that $j \leq t_1 - 1$. 
\end{proof}

\begin{lemma}\label{2MerOne}
For any fixed mergeable covers $C_1,C_2$ of a graph $G$ without loops, let $d_1 = 1$. Then $d_2 \leq t_2 - 1$. 
\end{lemma}
\begin{proof}
Consider the equation from Lemma \ref{2MergeEq}. If $d_1 = 1$ then $k = j \cdot d_2$. From the same lemma we also know that $GCD(j,k) = 1$, so necessarily $j = 1$ and $k = d_2$. Finally, we also have inequality $k \leq t_2 - 1$ which together give $d_2 \leq t_2 - 1$. 
\end{proof}

\subsection{Vertices $\psi(2)$ and $\psi(3)$ in $\mathcal{C}(A_n)$}

In this section, we explore the exact structure of terminals of graphs from $\mathcal{C}(A_n)$. Based on a number of simple properties we show that given large enough $n$, no graph from $\mathcal{C}(A_n)$ admits label 1 on any vertex from $\psi(2)$ and $\psi(3)$. 

For a vertex $v$ of a fixed labeling, let $\tau(v)$ denote the set of proper terminals of $v$. Lemma \ref{TermStruct} summarizes some basic observations. For its proof see the Appendix.

\begin{lemma}\label{TermStruct}
For any graph $G \in \mathcal{C}(A_n)$ such that $n \geq 39$, let us fix arbitrary vertices $v_i$ such that $v_i \in \psi(i)$ for values of $i$ from 2 to 6. Then for any integers $j,k$ such that $2 \leq j, 2 \leq k \leq 6$ all of the following holds: 

\begin{itemize}

\item[1.] Vertex $v_k$ has exactly $k$ proper terminals (and $k+1$ terminals in total), all of which have distinct labels. 

\item[2.] If $j \leq 3$, then the vertices $v_j$ and $v_k$ are mergeable. 

\item[3.] If labeling is minimal and either $\ell(v_2)$ or $\ell(v_3)$ equals to 1, then $\ell(v_k) \leq k$. 

\item[4.] There is exactly one proper terminal in the intersection of all $\tau(v_k)$, and this terminal has the highest label in the graph. 

\item[5.] If $j$ is not the highest label, then either $j \notin \tau(v_2)$ or $j \notin \tau(v_3)$. 

\item[6.] $\tau(v_k) \subset \tau(v_j)$ whenever $j \geq k + 2$. 

\item[7.] None of the chosen vertices is a proper terminal of any of the other chosen vertices and thus $\ell(v_j)+\ell(v_k) \in M$ for any $j \neq k$. 
\end{itemize}
\end{lemma}

For convenience, we extend the meaning of \textit{terminal}. We characterize possible labelings of graphs in terms of presence or absence of values in respect to the label-set of the graph. For a vertex $v$ we say an integer value $k$ is a (proper) \textit{terminal} for $v$ as a shortcut for the fact that there exists a vertex $w$, which is a (proper) terminal for $v$ and $\ell(w) = k$. We only deal with terminals of vertices from $\psi(2),...,\psi(6)$ whose all proper terminals are non-equivalent and thus have distinct labels. Hence the integer terminals and the vertex terminals are in one-to-one correspondence for these vertices. 

Let $G$ be a sum graph with a minimal labeling $\ell$. Let $M$ be the label-set of $G$ in respect to $\ell$ and let $v \in V(G)$ be such that $\ell(v) = 1$. According to Lemma \ref{SeqCover}, the vertex $v$ describes $M$ as a union of several distinct integer intervals separated by some values that are not elements of $M$. 

We say that an interval is \textit{long} if its first six labels and its last six labels do not intersect. In particular, its last six labels are strictly bigger than 6. The property of the long intervals we want to use is that for any $v$ such that $v \in \psi(i)$ for $i \leq 6$, any terminal among the last six elements of a long interval is always a proper terminal of $v$. Proofs of Lemma \ref{1noton2} and Lemma \ref{1noton3} are very similar, thus the proof of Lemma \ref{1noton2} is given in the Appendix.

\begin{lemma}\label{1noton2}
For any graph $G \in \mathcal{C}(A_n)$ such that $n \geq 39$, there is no labeling such that $\ell(v) = 1$ for any $v \in \psi(2)$. 
\end{lemma}

\begin{lemma}\label{1noton3}
For any graph $G \in \mathcal{C}(A_n)$ such that $n \geq 39$, there is no labeling such that $\ell(v) = 1$ for any $v \in \psi(3)$. 
\end{lemma}
\begin{proof}
Let us fix arbitrary vertices $v_i$ such that $v_i \in \psi(i)$ for values of $i$ from 2 to 6 and let us denote their labels as $\alpha := \ell(v_2), \beta := \ell(v_3), \gamma := \ell(v_4), \delta := \ell(v_5), \epsilon := \ell(v_6)$. For contradiction let $\beta = 1$. We use the observations from the previous Lemma \ref{TermStruct} and Lemma \ref{2MerOne} to reach contradiction. 

As $v_3$ has four terminals, $M$ composes of at most three non-trivial intervals with possible trivial interval $\{1\}$, let $X,Y,Z$ denote the intervals other than $\{1\}$ so that $X < Y < Z$. Let $k$ be the last label of $X$ or $Y$. 

As $v_2$ has only two proper terminals, $\alpha \leq 2$ and thus $\alpha = 2$. Since $k$ is a terminal for $v_3$, it is not a terminal for $v_2$. Together, we get that $k+2$ is a label and so both gaps between the intervals $X,Y,Z$ have size exactly 1. 

Let us focus on the interval $X$ and let $k$ from now on denote its last label. From $k$ being a terminal for $v_3$ we have that $k$ is also a terminal for $v_5$ and $v_6$. Thus there are at least three distinct non-labels strictly between $k$ and $k+7$, and so the third closest non-label following $k$ is at most $k+6$. As there are at most two of them in the gaps separating the three intervals, we have that the sum of lengths of $Y$ and $Z$ is at most 3 (together with 3 non-labels summing up to 6). From this we get that $X$ is long. 

Since $\alpha = 2$, there is one terminal for $v_2$ in $X$. If $Y$ or $Z$ has length two, then the first element is also a terminal for $v_2$. Together with the highest label in $M$ we would reach a contradiction with $v_2$ having only two proper terminals. 

We have that $Y = \{k+2\}$ and $Z = \{k+4\}$. The label $k$ is a terminal for $v_3$ and label $k-1$ is a terminal for $v_2$. Both are also terminals for $v_5$. Therefore, we need to set $\delta$ so that both $k+\delta$ and $k-1+\delta$ fall into $\{k+1,k+3,k+5\}$. But there is no such value and we reach a contradiction. 
\end{proof}

\subsection{Smallest Labels of $\mathcal{C}(A_n)$}

In this section we give limitations on ordering of labels in general case. We apply this together with our previously acquired knowledge to further limit the position of label 1 in graphs from $\mathcal{C}(A_n)$. 

\begin{lemma}\label{LabSmallK}
Let $G$ be a sum graph, let $v_1,v_2$ be equivalent vertices of $G$ with $k$ terminals. Let $\ell$ be any labeling of $G$ and $M$ the label-set associated with $\ell$. Then $\ell(v_1)$ is one of the $k$ smallest labels in $M$. Furthermore, if $\ell(v_1) = \ell(v_2)$ then $\ell(v_1)$ is one of the $k-1$ smallest labels in $M$. 
\end{lemma}
\begin{proof}
Consider $v_1$ and let us count the number of labels from $M$ not appearing on any neighbor of $v_1$. The vertex $v_1$ has exactly $k$ terminals including itself. Therefore there are at least $|M| - k$ labels appearing on the neighbors of $v_1$ and so the edges incident with $v_1$ carry at least $|M| - k$ distinct edge-numbers expressed as a sum of $\ell(v_1)$ and a positive label of one of the neighbors. Together we get that $M$ contains at least $|M| - k$ values strictly greater than $\ell(v_1)$. 

Assume that $\ell(v_1) = \ell(v_2)$. We can improve the previous argument by the fact that the label of $v_1$ is present on a neighbor of $v_1$ (namely $v_2$). This improves the bound on labels in $M$ strictly greater than $\ell(v_1)$ to at least $|M| - (k-1)$. 
\end{proof}

Consider graph $A_n$ for $n \geq 39$ with some labeling $\ell$, let $\alpha$ denote a label of some vertex from $\psi(2)$ and let $\beta$ denote a label of some vertex from $\psi(3)$. Recall Lemma \ref{TermStruct} giving explicit amounts of terminals of vertices from $\psi(2)$ and $\psi(3)$. As a corollary of the previous Lemma \ref{LabSmallK} we have that $\alpha$ is among the 3 smallest labels from $M$ and $\beta$ is among the 4 smallest labels from $M$. And as we already know from Lemma \ref{1noton2} and Lemma \ref{1noton3}, $\alpha,\beta \neq 1$. 

Recall Lemma \ref{2MergeEq}, which gives us limitations on mutual relations in between labels of mergeable vertices. As we know that $\psi(2)$ are mergeable with $\psi(3)$, we know that one of the following must hold: 
$\alpha = 2 \beta$, $2 \alpha = \beta$, $3 \alpha = \beta$ or $3 \alpha = 2 \beta$.

We use these facts to show that Lemma \ref{LabSmallK} can be applied to vertices from $\psi(2)$ and $\psi(3)$ in its stronger form, thus fully determining the smallest three labels in minimal labelings of graphs from $\mathcal{C}(A_n)$. 

\begin{lemma}\label{Lab2Unique}
In any minimal labeling $\ell$ of a graph $G \in \mathcal{C}(A_n)$ such that $n \geq 39$, $\ell(v_1) = \ell(v_2)$ for any $v_1,v_2 \in \psi(2)$. 
\end{lemma}
\begin{proof}
For contradiction, let $v_1,v_2 \in \psi(2)$ have distinct labels $\alpha_1, \alpha_2$. Without loss of generality, $\alpha_1 < \alpha_2$. We use the observations from Lemma \ref{TermStruct}. As $v_1$ is mergeable with $v_2$ and each has three terminals, Lemma \ref{2MergeEq} implies that necessarily $2\alpha_1 = \alpha_2$. 

Let $\beta$ denote a label of any vertex from $\psi(3)$. From mergeability, both values $\alpha_1,\alpha_2$ relate to $\beta$, from Lemma \ref{2MergeEq}. The only values for any of the two alphas are $2\beta, \frac{1}{2}\beta, \frac{1}{3}\beta,\frac{2}{3}\beta$. Since $2\alpha_1 = \alpha_2$, the only two suitable values are $\alpha_1 = \frac{1}{3}\beta$ and thus $\alpha_2 = \frac{2}{3}\beta$. 

Consider the $\alpha_2$-cover of the label-set $M$. Since both $\alpha_1$ and $\alpha_2$ are among the three smallest elements of $M$, according to Lemma \ref{LabSmallK}, we know that the two proper $\alpha_2$-sequences start with elements $1$ and $\alpha_1$, as these are smaller than $\alpha_2$. The value $\alpha_2$ is not an element of either of the two sequences, thus $\{\alpha_2\}$ forms an improper sequence and consequently $2 \alpha_2$ is not a label. From the last point of Lemma \ref{TermStruct}, the value $\alpha_1 + \beta = 4 \alpha_1 = 2 \alpha_2$ is a label and we reach a contradiction. 
\end{proof}

\begin{corollary}\label{AlpLeqBet}
For any minimal labeling $\ell$ of $A_n$, $\ell(v) < \ell(w)$ for any $v \in \psi(2)$ and $w \in \psi(3)$. 
\end{corollary}
\begin{proof}
The previous Lemma \ref{Lab2Unique} guarantees the additional condition of Lemma \ref{LabSmallK} on labels of vertices from $\psi(2)$, thus $\ell(v)$ is at most second smallest label in $A_n$. As 1 is also a label, $\ell(v)$ is exactly the second smallest and thus $1 \neq \ell(w) > \ell(v)$. 
\end{proof}

\begin{lemma}\label{Lab3Unique}
In any minimal labeling $\ell$ of a graph $G \in \mathcal{C}(A_n)$ such that $n \geq 39$, $\ell(v_1) = \ell(v_2)$ for any $v_1,v_2 \in \psi(3)$. 
\end{lemma}

The proof of Lemma \ref{Lab3Unique} is analogous to the proof of Lemma \ref{Lab2Unique} and is given in the Appendix. We are now ready to prove the last limitation on the placement of label 1 in minimal labelings of graphs from $\mathcal{C}(A_n)$ in order to exclude all possible minimal labelings. 

\begin{lemma}\label{1nothigh}
In any labeling $\ell$ of a graph $G \in \mathcal{C}(A_n)$ such that $n \geq 39$, if $v \in \psi(i)$ and $\ell(v) = 1$, then $i \leq 3$. 
\end{lemma}
\begin{proof}
Let $\alpha$ be a label of a vertex $v \in \psi(2)$ and let $\beta$ be a label of a vertex $w \in \psi(3)$. From Lemma \ref{Lab2Unique} and Lemma \ref{Lab3Unique}, we know that these values are uniquely determined by $\ell$. For contradiction, let $\alpha,\beta > 1$. Let $M$ denote the label-set of $G$. From the previous Lemma \ref{Lab2Unique} and Lemma \ref{Lab3Unique}, we have the additional conditions to apply Lemma \ref{LabSmallK} to $v$ and $w$ in its stronger form. Together we have that $\alpha < \beta$ (corollary of Lemma \ref{Lab2Unique}) and labels $1,\alpha,\beta$ are the three smallest labels in $M$. 

Let $x$ be a vertex such that $\ell(x) = 1$. Suppose $x$ is not a terminal for $v$. Then $\alpha + 1$ is a label. Since $\beta$ is the first label greater than $\alpha$, we have $\beta = \alpha + 1$. From the fact that $\alpha > 1$ and $\beta$ is not a multiple of $\alpha$ and vice versa, the mergeability of $v$ and $w$ leaves only one possible relation, $3 \alpha = 2 \beta$. We conclude that $\alpha = 2$ and $\beta = 3$. Consider any $\beta$-sequence, any of its consecutive elements fall into distinct $\alpha$-sequences. As there exists at least one $\beta$-sequence with at least 8 elements, the two proper $\alpha$-sequences must overlap to satisfy the last 4 elements, none of which can fall into the improper $\alpha$-sequence. Let $k$ be a label of a proper terminal of $v$. As $x$ is induced by a number of size at least 4, any proper terminal label of $v$ is also a terminal label of $x$. Thus, $k+1$ is not a label. This means that the other $\alpha$-sequence not terminating in $k$ cannot extend over $k+1$ and thus, either ends before $k$ or begins after $k+1$. Applying the same argument to the other sequence, we get that the two $\alpha$-sequences must not overlap as none can extend over the last element of the other. This is a contradiction and $x$ must be a proper terminal for $v$. 

Since $1$ is a proper terminal of $v$, the $\alpha$-cover of $M$ has only one non-trivial sequence with the first element $k$ (for some yet unknown integer $k$). Values $\beta$ and $\beta + \alpha$ are elements of $M$ with difference $\alpha$, and thus are consecutive elements of the only proper $\alpha$-sequence. Since the only terminal shared between $v$ and $w$ has the biggest label in $M$, $x$ is not a terminal of $w$ and thus $\beta+1$ is a label and must belong to the proper $\alpha$-sequence (as $\beta > \alpha$). Together, we have that the proper $\alpha$-sequence contains elements $\beta$ and $\beta+1$; thus, the difference $\alpha$ must equal to 1, which is a contradiction. 
\end{proof}

\subsection{Results}


\begin{proof}[of Theorem \ref{no1onoinj}]
Let $G$ be any graph from $\mathcal{C}(A_n)$ where $n \geq 39$ and let $\ell$ be any labeling of $G$. For contradiction, let $\ell(v) = 1$ for some vertex $v$. Clearly $v \in \psi(i)$ for some integer $i$, $1 < i \leq n+2$. As shown by Lemma \ref{1noton2} and Lemma \ref{1noton3}, $i > 3$. The Lemma \ref{1nothigh}, based on the previously mentioned lemmas, shows the complementary fact that $i \leq 3$. This is of course a contradiction. 
\end{proof}

The constant 39 is an artifact of used methods and may in fact be much smaller. While minimal labeling exists for graphs from $\mathcal{C}(A_6)$, based on a computer search we conjecture that there is in fact no minimal labeling for any graph from $\mathcal{C}(A_n)$ for any $n \geq 7$.  


\begin{proof} [of Theorem \ref{thm:loops}]
Let $G$ be a graph $A_n$ where $n \geq 39$. For contradiction let $G$ admit a minimal labeling. We replace each vertex of $G$ with loop by a clique, obtaining a graph $H \in \mathcal{C}(G)$, and assign all the new vertices from each new clique the label of the original vertex this clique replaces. We have a graph $H$ with sum labeling using the same labels as the labeling of $G$. Hence if $G$ admits a minimal sum labeling then $H$ must also admit a minimal sum labeling, which is a contradiction with Theorem \ref{no1onoinj}. 
\end{proof}

\section{Conclusion}

We have shown that the set $\{2,3,..,n-1,n,n+2\}$ for $n\ge 39$ induces a family of sum graphs with loops which admit no minimal labeling. Furthermore the loops can be replaced by a cliques of sizes at least two and we obtain an infinite family of non-injective sum graphs (without loops) which also admit no minimal labeling.

The constant 39 is an artefact of used methods and it might be much smaller. While minimal labeling exists for graphs from $\mathcal{C}(A_6)$, based on a computer search, we put forward the following conjecture: 

\begin{conjecture}
Let $G \in \mathcal{C}(A_n)$ such that $n \geq 7$, then $G$ allows no minimal labeling.  
\end{conjecture}

Further computer experiments indicate that it is not necessary to restrict to the omission of the second-last element from the sequence $\{2,3,..,n-1,n,n+1,n+2\}$. Thus, we put forward the following conjecture regarding sum graphs with loops: 

\begin{conjecture}
Let $G$ be a gap-graph of size $k$, where $k \geq 10$, with gap $i$ such that $3 < i < k$, then $G$ admits no minimal labeling. 
\end{conjecture}

\section{Acknowledgements}
This paper is the output of the 2016 Problem~Seminar. We would like to thank
Jan~Kratochv\'{i}l and Ji\v{r}\'{i}~Fiala for their guidance, help and tea.

\newpage

\section*{Appendix}

\begin{proof}[of Theorem \ref{No1wlEx}]
Consider a sum graph induced by the set $M = \{2,3,4,6,7\}$. For contradiction, assume there is a labeling $\{a,b,c,d,e\}$ using label 1. Let $v_i$ denote the unique vertex from $\psi(i)$ for every $i \in M$. As vertices $v_2,v_3,v_4$ are not equivalent, the labels $a,b,c$ need to be distinct. If we take the greatest of the three labels, it belongs to a vertex incident with two edges with distinct edge-numbers that need to be guaranteed by labels $d,e$.  We conclude that $a,b,c < d,e$. Therefore, the only values that can be equal to 1 are $a,b$ and $c$. Without loss of generality we can also assume $a<b$ and $d<e$. Let us go through the remaining cases:
\begin{enumerate}
\item $a=1$: since there is a loop incident with the vertex $v_2$, there needs to be a vertex labeled 2 and from $b,c<d,e$ we see it has to be either $b$ or $c$.

\begin{enumerate}
\item $b=2$: as there is no edge $v_2v_3$, no label can be equal to 3. Since $v_2$ has a loop with edge-number 4, we have $c=4$. Because of edges $v_2v_4$ and $v_3v_4$ we have $d=5$ and $e=6$. This is a contradiction because there is no edge $v_2v_6$.
\item $c=2$: since $v_2v_4$ is an edge and $b<d,e$, it holds $b=3$. From the edge $v_3v_4$ and the loop on $v_3$, we conclude $d=5$ and $e=6$, which is a contradiction because there is no edge $v_2v_6$.
\end{enumerate}

\item $c=1$: from the presence of the edge $v_2v_4$ and $b<d,e$ we have $b=a+1$. From the edge $v_3v_4$ and loops on $v_2$ and $v_3$ we have three edge-numbers $2a,2a+2$ and $a+2$ that are all distinct from labels $a,b$ and $c$ because $a>1$. Two of these edge-numbers need to be equal, as there are only two vertices to guarantee all three. The only option is $2a=a+2$, which gives us $a=2$. As there is the edge $v_2v_4$ and $b<d,e$, it follows $b=3$ and consequently $d=4$ and $e=6$. This is a contradiction because there is no edge between $v_2v_6$.
\end{enumerate}
\end{proof}

\begin{proof}[of Lemma \ref{mergeability}]
Suppose that the number of non-equivalent vertices of $G$ is at least $2 t_i t_j + 1$. As there are at most $t_i$ sequences in $C_i$ and at most $t_j$ sequences in $C_j$, there are at most $t_i \cdot t_j$ possible intersections. Since all labels are covered by both $C_i$ and $C_j$, each label falls to one of the intersections. The claim holds via Pigeonhole principle. 

Notice that if $C_i$ contains the maximum number of distinct sequences, from Lemma \ref{SeqCover} we know that one of the sequences is a singleton sequence $(d_i)$. Clearly this sequence can have at most one non-zero intersection with any sequence from $C_j$, and it can be of size at most one. From a symmetrical argument for $C_j$ we get that there are actually at most $(t_i - 1) \cdot (t_j - 1)$ possible intersections that can contain more than one element and two extra intersections containing at most one element each. Thus, the lemma holds via Pigeonhole principle. 
\end{proof}

\begin{proof}[of Lemma \ref{TermStruct}]
The first point is an observation based on the structure of $A_n$ resp. $\mathcal{C}(A_n)$. The claim holds for $A_n$. As all the chosen vertices are induced by small values, the involved proper terminals are loopless pair-wise non-equivalent vertices (induced by high values) in $A_n$, so the number of proper terminals is the same in $G$ and all the proper terminals are non-equivalent in $\mathcal{C}(A_n)$. 

To prove the second point, we first note that all vertices in $A_n$ are non-equivalent. Thus, the same holds for any two vertices of $G$ induced by distinct integer values. From this we have that there are at least $n$ non-equivalent vertices in $G$. From the definition of mergeability, it suffices to show that $v_3$ and $v_6$ are mergeable, as the condition for these vertices is the most restrictive. Applying the first point, we get that the condition explicitly states $n \geq 2 \cdot 3 \cdot 6 + 3 = 39$, which is satisfied. 

The third point is a direct consequence of the first two points and Lemma \ref{2MerOne}. 

The fact that the vertex with the highest label is a terminal for every vertex follows as the highest label is necessarily the last element of a sequence in any cover of the label-set. On the other hand, there can be at most one universally shared element, as the proper terminals of $v_2$ are induced by values $n-1$ and $n+2$, while the proper terminals of $v_3$ are induced by the values $n-2$, $n$ and $n+2$. None of these values induces multiple vertices in $A_n$, which together with the previous observation proves points 4 and 5. 

To prove the sixth point, let us consider $\tau(v_k)$ in $A_n$. Whenever vertex $w$ induced by an integer $i$ is a terminal for $v_k$ then either $i+k = n+1$ or $i+k > n+2$. In both cases it also holds that $i+(k+2) > n+2$ and thus $w$ is a terminal for $v_{k+2}$. 

The final point holds as all vertices induced by values 2 to 6 are connected in $A_n$ and thus also in $G$. 
\end{proof}

\begin{proof}[of Lemma \ref{1noton2}]
Let us fix arbitrary vertices $v_i$ such that $v_i \in \psi(i)$ for values of $i$ from 2 to 6 and let us denote their labels as $\alpha := \ell(v_2), \beta := \ell(v_3), \gamma := \ell(v_4), \delta := \ell(v_5), \epsilon := \ell(v_6)$. For contradiction let $\alpha = 1$. We use the observation from the previous Lemma \ref{TermStruct}  and Lemma \ref{2MerOne} to reach contradiction. 

As $v_2$ has three terminals, $M$ composes of at most two non-trivial intervals with possible trivial interval $\{1\}$, let $X,Y$ denote the intervals other than $\{1\}$ so that $X < Y$ and let $k$ be the last label of $X$. 

Label $k$ is a terminal for $v_2$, thus it is also a terminal for $v_4,v_5$ and $v_6$. On the other hand, $k$ is not a terminal for $v_3$ as its not the highest label in $M$. 

Consider label of $v_3$ (denoted $\beta$). As $k$ is a not terminal for $v_3$, the element $k + \beta \leq k + 3$ is an element of $Y$, so there are at most two elements in the gap between the intervals $X$ and $Y$. Furthermore, if the gap contains two elements, then $\beta$ must be exactly 3. 

Consider labels of $v_4$ and $v_5$ (denoted $\gamma$ resp. $\delta$). As $k$ is a terminal for both of them (and $v_2$), we have that $k+\gamma \leq k+4$, $k+\delta \leq k+5$ and $k+\alpha = k+1$ are three distinct non-labels. As there are at most two of them in between $X$ and $Y$, we get that the sum of the lengts of the gap and the interval $Y$ is at most 3. Thus $Y$ contains at most two elements and $X$ is long. 

Consider the gap to have size 2. We have $\beta = 3$ and since the first interval is long, we have that $v_3$ has two terminals $i,j$ such that $i+\beta,j+\beta$ fall into the gap. As $v_3$ has three proper terminals, $Y$ has at most one element; as otherwise, it would hold two additional terminals for $v_3$. We know that all terminals of $v_2$ and $v_3$ are also terminals of $v_5$. Since $\beta = 3$, we have that $k,k-1,k-2$ are terminals of $v_5$. Thus, $\delta$ is a value such that none of the values $k+\delta,k-1+\delta,k-2+\delta$ are equal to $k$ or $k+3$. As $\delta \leq 5$, there is no possible value of $\delta$ and we reach a contradiction. This shows that the size of the gap is 1. 

Consider the interval $Y$ to have size 1. The interval $Y$ must consist of a single element $k+2$. As $k$ is not a terminal for $v_3$, necessarily $k + \beta$ is the only possible label $k + 2$ and thus $\beta = 2$. Notice that the only proper terminals of $v_3$ are now $k-1$ and $k+1$. However $v_3$ has three proper terminals, which is a contradiction and we conclude that $|Y| = 2$. 

Since $k$ is a terminal for $v_2$, it is also a terminal for $v_4$, thus $k + \gamma$ is a non-label and $k+2 \leq k+\gamma \leq k+4$. There is only one such non-label, $k+4$ thus $\gamma = 4$. We have that $(k + 1) - \gamma = k-3$ is a terminal of $v_4$ and thus also a terminal of $v_6$. The label $k-3+\epsilon$ is a non-label. Since $\epsilon \neq \gamma$, the lowest possible non-label is $k+4$, meaning that $\epsilon$ must be at least 7, however $\epsilon \leq 6$ and we reach a contradiction with the existence of the assumed labeling $\ell$. 
\end{proof}

\begin{proof}[of Lemma \ref{Lab3Unique}]
We proceed analogously as in the proof of Lemma \ref{Lab2Unique}. For contradiction, let $v_1,v_2 \in \psi(3)$ have distinct labels $\beta_1, \beta_2$. Without loss of generality, $\beta_1 < \beta_2$. As $v_1$ is mergeable with $v_2$ and each has four terminals, Lemma \ref{2MergeEq} implies that one of the equations $2\beta_1 = \beta_2; 3\beta_1 = \beta_2; 3\beta_1 = 2\beta_2$ holds. 

Let $\alpha$ denote a label of any vertex from $\psi(2)$. From mergeability, both values $\beta_1,\beta_2$ relate to $\alpha$, according to Lemma \ref{2MergeEq}. Since from the Corollary \ref{AlpLeqBet} we have $\alpha < \beta_1,\beta_2$, the only possible values of $\beta_1,\beta_2$ are $2\alpha,3\alpha,\frac{3}{2}\alpha$. The only two solutions are $\beta_1 = \frac{3}{2}\alpha; \beta_2 = 3\alpha$ and $\beta_1 = 2\alpha; \beta = 3\alpha$. 

Suppose the former solution holds. All values 1,$\alpha$,$\beta_1$ are strictly smaller than $\beta_2$ and since $\beta_2$ belongs to the four smallest labels, there are no further elements smaller than $\beta_2$. From the last point of Lemma \ref{TermStruct}, we have that $\alpha + \beta_1$ is a label. From the relations, we have that $\beta_1 < \alpha + \beta_1 < 3 \alpha = \beta_2$. Thus, $\alpha + \beta_1$ is another distinct element strictly smaller than $\beta_2$ and we reach a contradiction. 

Thus, the latter solution must apply. Consider $\beta_2$-cover, it must consist of three proper sequences with first elements $1,\alpha$ and $\beta_1$. As $\beta_1 = 2\alpha$, the only labels that are not multiples of $\alpha$ must be of form $1+i\beta_2 = 1+3i\alpha$ for non-negative integer values $i$. Consider the $\alpha$-cover, it must consist of two proper sequences with first elements $1$ and $\alpha$, as $2\alpha = \beta_1$ is a label. Thus, if there is more than one label that is not multiple of $\alpha$, then there is label $1+\alpha$. Together with the previous, $1+\alpha$ is not a label and thus, there is at most one label that is not a multiple of $\alpha$, label 1. This means that 1 is a terminal for vertex from $\psi(2)$ and simultaneously for a $v_2$. This is a contradiction as $\psi(2)$ and $\psi(3)$ can only share the highest label. 
\end{proof}

\end{document}